\documentclass[11pt]{article}

\usepackage[bookmarks=false]{hyperref}

\def\colorful{1}

\usepackage{amsthm,amsfonts,amsmath,amssymb,epsfig,color,float,graphicx,verbatim}
\usepackage{multirow}

\newif\ifhyper\IfFileExists{hyperref.sty}{\hypertrue}{\hyperfalse}
\hyperfalse
\hypertrue
\ifhyper\usepackage{hyperref}\fi

\usepackage{enumitem}
\usepackage[capitalize]{cleveref} 

\makeatletter
\renewcommand{\section}{\@startsection{section}{1}{0pt}{-12pt}{5pt}{\large\bf}}
\renewcommand{\subsection}{\@startsection{subsection}{2}{0pt}{-12pt}{-5pt}{\normalsize\bf}}
\renewcommand{\subsubsection}{\@startsection{subsubsection}{3}{0pt}{-12pt}{-5pt}{\normalsize\bf}}
\makeatother

\usepackage{framed}
\usepackage{nicefrac}

\def\nnewcolor{1}
\ifnum\nnewcolor=1

\fi
\ifnum\nnewcolor=0

\fi

\ifnum\colorful=1

\else

\fi

\newtheorem{theorem}{Theorem}

\newtheorem{lemma}[theorem]{Lemma}
\newtheorem{proposition}[theorem]{Proposition}

\newtheorem{claim}[theorem]{Claim}
\newtheorem{fact}[theorem]{Fact}
\newtheorem{obs}[theorem]{Observation}

\theoremstyle{definition}
\newtheorem{definition}[theorem]{Definition}

\newcommand{\R}{\mathbb{R}}
\newcommand{\Z}{\mathbb{Z}}
\newcommand{\E}{\mathbb{E}}

\newcommand{\gs}{\geqslant}
\newcommand{\ls}{\leqslant}

\newcommand{\var}{\text{Var}}

\newcommand{\width}{\mathbf{width}}
\newcommand{\discr}{\mathbf{Discr}}

\newcommand{\ignore}[1]{}

\newcommand{\eps}{\epsilon}

\newcommand{\Ak}{\mathcal{A}_k}

\newcommand{\Var}{\mathop{\textnormal{Var}}\nolimits}

\newcommand{\Poi}{\mathop{\textnormal{Poi}}\nolimits}

\newcommand{\eqdef}{\stackrel{{\mathrm {\footnotesize def}}}{=}}

\newcommand{\littlesum}{\mathop{\textstyle \sum}}

\newenvironment{algorithm}[1][\  ] %
{ \rm
\begin{tabbing}
....\=.....\=.....\=.....\=.....\=  \+ \kill
} %
{\end{tabbing} }

\floatstyle{ruled}
\newfloat{Algorithm}{H}{alg}
\newfloat{Subroutine}{H}{sub}

{
\begin{minipage}{1.0\linewidth} \begin{algorithm} %
} { \end{algorithm} \end{minipage} }

\title{Optimal Algorithms and Lower Bounds for \\ Testing Closeness of Structured Distributions}

\author{
Ilias Diakonikolas\thanks{Supported by EPSRC grant EP/L021749/1 and a Marie Curie Career Integration grant.}\\
University of Edinburgh\\
{\tt ilias.d@ed.ac.uk}.\\
\and
Daniel M. Kane\\
University of California, San Diego\thanks{Some of this work was performed while visiting the University of Edinburgh.}\\
{\tt dakane@cs.ucsd.edu}.\\
\and
Vladimir Nikishkin\thanks{Supported by a University of Edinburgh PCD Scholarship.}\\
University of Edinburgh\\
{\tt v.nikishkin@sms.ed.ac.uk}.
}

\begin{document}

\maketitle

\thispagestyle{empty}

\begin{abstract}
We give a general unified method that can be used for $L_1$ {\em closeness testing} of a wide range of univariate structured distribution families.
More specifically, we design a sample optimal and computationally efficient algorithm for testing
the equivalence of two unknown (potentially arbitrary) univariate distributions under the $\mathcal{A}_k$-distance metric:
Given sample access to distributions with density functions $p, q: I \to \R$, we want to distinguish
between the cases that $p=q$ and $\|p-q\|_{\mathcal{A}_k} \ge \eps$ with probability at least $2/3$.
We show that for any $k \ge 2, \eps>0$, the {\em optimal} sample complexity of the  $\mathcal{A}_k$-closeness testing
problem is $\Theta(\max\{ k^{4/5}/\eps^{6/5}, k^{1/2}/\eps^2 \})$.
This is the first $o(k)$ sample algorithm for this problem, and yields
new, simple $L_1$ closeness testers, in most cases with optimal sample complexity,
for broad classes of structured distributions.
\end{abstract}

\thispagestyle{empty}
\setcounter{page}{0}

\newpage

\section{Introduction}  \label{sec:intro}

We study the problem of closeness testing (equivalence testing) between two unknown probability distributions.
Given independent samples from a pair of distributions $p, q$, we want to
determine whether the two distributions are the same versus significantly different.
This is a classical problem in statistical hypothesis testing~\cite{NeymanP, lehmann2005testing} that
has received considerable attention by the TCS community
in the framework of {\em property testing}~\cite{RS96, GGR98}:
given sample access to distributions $p, q$,
and a parameter $\eps>0$, we want to distinguish between the cases that $p$ and $q$ are identical versus
$\eps$-far from each other in $L_1$ norm (statistical distance).
Previous work on this problem focused on characterizing the sample size needed to test the identity of two arbitrary
distributions of a given support size~\cite{BFR+:00, CDVV14}. It is now known that the optimal sample complexity
(and running time) of this problem for distributions with support of size $n$ is $\Theta(\max\{n^{2/3}/\eps^{4/3}, n^{1/2}/\eps^2 \})$.

The aforementioned sample complexity characterizes worst-case instances,
and one might hope that drastically better results can be obtained for most natural settings,
in particular when the underlying distributions are known a priori to have some ``nice structure''.
In this work, we focus on the problem of testing closeness for {\em structured} distributions.
Let $\mathcal{C}$ be a family over univariate distributions.
The problem of {\em closeness testing for $\mathcal{C}$} is the following:
Given sample access to two unknown distribution $p, q \in \mathcal{C}$,
we want to distinguish between the case that $p = q$ versus $\|p-q\|_1 \ge \eps.$
Note that the sample complexity of this testing problem depends on the underlying class
$\mathcal{C}$, and we are interested in obtaining efficient algorithms that are {\em sample optimal} for $\mathcal{C}$.

We give a general algorithm that can be used for $L_1$ closeness testing of a wide range of structured distribution families.
More specifically, we give a sample optimal and computationally efficient algorithm for testing the identity of two unknown (potentially arbitrary)
distributions $p, q$ under a different metric between distributions -- the so called $\mathcal{A}_k$-distance (see Section~\ref{sec:results} for a formal definition).
Here, $k$ is a positive integer that intuitively captures the number of ``crossings'' between the probability density functions $p, q$.

Our main result (see Theorem~\ref{thm:main}) says the following:
{\em For any $k \in \Z_+, \eps>0$, and sample access to arbitrary univariate distributions $p, q$,
there exists a closeness testing algorithm under the $\mathcal{A}_k$-distance
using $O(\max\{ k^{4/5}/\eps^{6/5}, k^{1/2}/\eps^2 \})$ samples.
Moreover, this bound is information-theoretically optimal.}
We remark that our $\mathcal{A}_k$-testing algorithm applies to {\em any} pair of univariate distributions (over both continuous
and discrete domains). The main idea in using this general algorithm for testing closeness of structured distributions in $L_1$ distance is this:
if the underlying distributions $p, q$ belong to a structured distribution family $\mathcal{C}$, we can use the
$\mathcal{A}_k$-distance as a proxy for the $L_1$ distance (for an appropriate value of the parameter $k$),
and thus obtain an $L_1$ closeness tester for $\mathcal{C}$.

We note that $\mathcal{A}_k$-distance between distributions
has been recently used to obtain sample optimal efficient algorithms for {\em learning} structured distributions~\cite{CDSS14, ADLS15},
and for testing the identity of a structured distribution against an {\em explicitly known} distribution~\cite{DKN15a} (e.g., uniformity testing).
In both these settings, the sample complexity of the corresponding problem (learning/identity testing) with respect to the
$\mathcal{A}_k$-distance is identified with the sample complexity of the problem under the $L_1$ distance {\em for distributions of support $k$}.
More specifically,  the sample complexity of  learning an unknown univariate distribution (over a continuous or discrete domain) 
up to $\mathcal{A}_k$-distance $\eps$ is $\Theta(k/\eps^2)$~\cite{CDSS14} (independent of the domain size),
which is exactly the sample complexity of learning a {\em discrete distribution with support size $k$} up to $L_1$ error $\eps$. Similarly,
the sample complexity of uniformity testing of a univariate distribution 
(over a continuous or discrete domain)   up to $\mathcal{A}_k$-distance $\eps$ is  $\Theta(k^{1/2}/\eps^2)$~\cite{DKN15a} (again, independent of the domain size),
which is identical to the sample complexity of uniformity testing of a discrete distribution with support size $k$
up to $L_1$ error $\eps$~\cite{Paninski:08}.

Rather surprisingly, this analogy is {\em provably false} for the closeness testing problem: we prove that
the sample complexity of the $\mathcal{A}_k$ closeness testing problem  is $\Theta(\max\{ k^{4/5}/\eps^{6/5}, k^{1/2}/\eps^2 \})$, while
$L_1$ closeness testing between distributions of support $k$ can be achieved with $O(\max\{ k^{2/3}/\eps^{4/3}, k^{1/2}/\eps^2 \})$ samples~\cite{CDVV14}.
More specifically, our upper bound for $\mathcal{A}_k$ closeness testing problem applies for all univariate probability distributions (both continuous and discrete).
Our matching information--theoretic lower bound holds for continuous distributions, or discrete distributions of support size $n$ sufficiently large 
as a function of $k$, which is the most interesting regime for our applications. 

\vspace{-0.2cm}

\subsection{Related and Prior Work} \label{ssec:literature} In this subsection we review the related literature and compare
our results with previous work.

\smallskip

\noindent {\bf Distribution Property Testing} Testing properties of distributions~\cite{BFR+:00, Batu13}
has developed into a mature research area within theoretical computer science.
The paradigmatic problem in this field is the following: given sample access to one or more unknown probability distributions,
determine whether they satisfy some global property or are ``far'' from satisfying the property.
The goal is to obtain an algorithm for this task that is both statistically and computationally efficient, i.e., 
an algorithm with (information--theoretically) optimal sample size and polynomial runtime.
See~\cite{GR00, BFR+:00, BFFKRW:01, Batu01, BDKR:02, BKR:04,  Paninski:08, PV11sicomp, ValiantValiant:11,
DDSVV13, DJOP11, LRR11, ILR12, CDVV14, VV14, DKN15a} for a sample of works,
and~\cite{Rub12} for a survey.

\smallskip

\noindent {\bf Shape Restricted Estimation} 
Statistical estimation under shape restrictions -- i.e., 
inference about a probability distribution under the constraint 
that its probability density function satisfies certain qualitative properties --
is a classical topic in statistics~\cite{BBBB:72}.
Various structural restrictions have been studied in the literature, starting from
monotonicity, unimodality, convexity, and concavity~\cite{Grenander:56, Brunk:58, PrakasaRao:69, Wegman:70, HansonP:76, Groeneboom:85, Birge:87, Birge:87b,
Fougeres:97,ChanTong:04,JW:09},
and more recently focusing on structural restrictions such as log-concavity and $k$-monotonicity
\cite{BW07aos, DumbgenRufibach:09, BRW:09aos, GW09sc, BW10sn, KoenkerM:10aos, Walther09, DossW13, ChenSam13, KimSam14, BalDoss14, HW15}.
The reader is referred to~\cite{GJ:14} for a recent book on the topic.

\smallskip

\noindent {\bf Comparison with Prior Work}
Chan, Diakonikolas, Servedio, and Sun~\cite{CDSS14} proposed a general approach to $L_1$ learn univariate probability distributions
whose densities are well approximated by piecewise polynomials. They designed an efficient
agnostic learning algorithm for piecewise polynomial distributions, and as a corollary obtained efficient
learners for various families of structured distributions. 
The approach of \cite{CDSS14} uses the $\mathcal{A}_k$ distance metric between distributions,
but is otherwise orthogonal to ours.  Batu {\em et al.}~\cite{BKR:04} gave 
algorithms for closeness testing between two monotone distributions with sample complexity $O(\log^3 n).$
Subsequently, Daskalakis {\em et al.}~\cite{DDSVV13} improved and generalized this result 
to $t$-modal distributions, obtaining a closeness tester  with
sample complexity $O((t \log(n))^{2/3}/\eps^{{8}/3}+t^2/\eps^4)$. We remark that the approach of~\cite{DDSVV13}
inherently yields an algorithm with sample complexity $\Omega(t)$, which is sub-optimal.

The main ideas underlying this work are very different from those of ~\cite{DDSVV13} and~\cite{DKN15a}.
The approach of~\cite{DDSVV13} involves constructing an adaptive interval decomposition of the domain
followed by an application of a (known) closeness tester to the ``reduced'' distributions over those intervals.
This approach incurs an extraneous term in the sample complexity, that is needed
to construct the appropriate decomposition.
The approach of~\cite{DKN15a} considers several oblivious
interval decompositions of the domain (i.e., without drawing any samples) and applies
a ``reduced'' identity tester for each such decomposition. This idea yields sample--optimal bounds 
for ${\cal A}_k$ identity testing against a {\em known} distribution.
However, it crucially exploits the knowledge of the explicit distribution, 
and unfortunately fails in the setting where both distributions are unknown. 
We elaborate on these points in Section~\ref{ssec:techniques}.




\section{Our Results and Techniques} \label{sec:results}

\subsection{Basic Definitions} \label{ssec:defs}
We will use $p, q$ to denote the probability density functions (or probability mass functions)
of our distributions. If $p$ is discrete over support $[n]: = \{1, \ldots, n\}$, we denote
by $p_i$ the probability of element $i$ in the distribution.
For two discrete distributions $p, q$, their $L_1$ and $L_2$ distances are
$\|p -q \|_1 = \sum_{i=1}^n |p_i - q_i|$ and $\|p-q\|_2 = \sqrt{\sum_{i=1}^n (p_i - q_i)^2}$.
For $I \subseteq \R$ and density functions $p, q: I \to \R_+$, we have $\|p -q \|_1 = \int_I |p(x)-q(x)| dx$.

Fix a partition of the domain $I$ into disjoint intervals
$\mathcal{I} :=  (I_i)_{i=1}^{\ell}.$ For such a partition $\mathcal{I}$,
the {\em reduced distribution} $p_r^{\mathcal{I}}$ corresponding to $p$ and $\mathcal{I}$ is the discrete distribution over $[\ell]$
that assigns the $i$-th ``point'' the mass that $p$ assigns to the
interval $I_i$; i.e., for $i \in [\ell]$, $p_r^{\mathcal{I}} (i) = p(I_i)$.
Let $\mathfrak{J}_k$ be the collection
of all partitions of the domain $I$ into $k$ intervals. For $p, q: I \to \R_+$ and $k \in \Z_+$,
we define the $\mathcal{A}_k$-distance between $p$ and $q$ by
$$\|p-q\|_{\mathcal{A}_k} \eqdef \max_{\mathcal{I} = (I_i)_{i=1}^{k} \in \mathfrak{J}_k} \littlesum_{i=1}^k |p(I_i) - q(I_i)|
= \max_{\mathcal{I}  \in \mathfrak{J}_k} \| p_r^{\mathcal{I} } - q_r^{\mathcal{I} } \|_1.$$

\subsection{Our Results} \label{ssec:results}


Our main result is an optimal algorithm and a matching information--theoretic lower bound
for the problem of testing the equivalence between two unknown univariate distributions under the
$\mathcal{A}_k$ distance metric:
\begin{theorem}[Main] \label{thm:main}
Given $\eps>0$, an integer $k \ge 2$, and sample access to two distributions with probability density functions $p, q: [0, 1] \to \R_+$,
there is a computationally efficient algorithm which uses $O(\max\{ k^{4/5}/\eps^{6/5}, k^{1/2}/\eps^2 \})$ samples from $p, q$,
and with probability at least $2/3$ distinguishes whether $q = p$ versus $\|q-p\|_{{\cal A}_k} \ge \eps$.
Additionally, ${\Omega}(\max\{ k^{4/5}/\eps^{6/5}, k^{1/2}/\eps^2 \})$ samples are information-theoretically necessary for this task.
\end{theorem}

Note that Theorem~\ref{thm:main} applies to arbitrary univariate distributions (over both continuous and discrete domains). 
In particular, the sample complexity of the algorithm does not depend on the support size of the underlying distributions.
We believe that the notion of testing under the $\mathcal{A}_k$ distance is very natural, and well suited
for (arbitrary) continuous distributions, where the notion of $L_1$ testing is (provably) impossible.

As a corollary of Theorem~\ref{thm:main}, we obtain sample--optimal algorithms for the $L_1$ closeness testing of various
structured distribution families ${\mathcal{C}}$ in a unified way. 
The basic idea is to use the $\mathcal{A}_k$ distance as a ``proxy'' for the $L_1$ distance
for an appropriate value of $k$ that depends on ${\mathcal{C}}$ and $\eps$.
We have the following simple fact:

\begin{fact} \label{fact:simple}
For a univariate distribution family ${\mathcal{C}}$ and $\eps>0$, let $k= {k({\mathcal{C}}, \eps)}$ 
be the smallest integer such that for any $f_1, f_2 \in {\mathcal{C}}$ it holds that
$\|f_1-f_2\|_1 \le \|f_1-f_2\|_{{\mathcal A}_k} + \eps/2$. Then there exists an $L_1$ closeness testing algorithm 
for $\mathcal{C}$ using $O(\max\{ k^{4/5}/\eps^{6/5}, k^{1/2}/\eps^2 \})$ samples.
\end{fact}

Indeed, given sample access to $q, p \in {\mathcal{C}}$,
we apply the $\mathcal{A}_k$-closeness testing algorithm of Theorem~\ref{thm:main} for the value of $k$ in the statement of the fact,
and error $\eps' = \eps/2$. If $q = p$, the algorithm will output ``YES'' with probability at least $2/3$. If $\|q - p\|_1 \ge \eps$, then by the condition of
Fact~\ref{fact:simple} we have that $\|q-p\|_{\mathcal{A}_k} \ge \eps'$, and the algorithm will output ``NO'' with probability at least $2/3$.

We remark that the value of $k$ in Fact~\ref{fact:simple}  is a natural complexity measure 
for the difference between two probability density functions in the class ${\mathcal{C}}$.
It follows from the definition of the ${\mathcal A}_k$ distance that this value corresponds 
to the number of ``essential'' crossings between $f_1$ and $f_2$ --  i.e.,
the number of crossings between the functions $f_1$ and $f_2$ 
that significantly affect their $L_1$ distance. Intuitively, the number of essential crossings -- as opposed
to the domain size -- is, in some sense, the ``right'' parameter to characterize 
the sample complexity of $L_1$ closeness testing for ${\mathcal{C}}$.

The upper bound implied by the above fact is information-theoretically optimal for a wide range of structured distribution classes ${\mathcal{C}}$.
In particular, our bounds apply to all the structured distribution families considered in
~\cite{CDSS14, DKN15a, ADLS15} including (arbitrary mixtures of) $t$-flat (i.e., piecewise constant with $t$ pieces), $t$-piecewise degree-$d$ polynomials,
$t$-monotone, monotone hazard rate, and  log-concave distributions.  For $t$-flat distributions we obtain an $L_1$ closeness testing algorithm
that uses $O(\max\{ t^{4/5}/\eps^{6/5}, t^{1/2}/\eps^2 \})$ samples, which is the first $o(t)$ sample algorithm for the problem.
For log-concave distributions, we obtain a sample size of $O(\eps^{-9/4})$ matching the information--theoretic lower bound even for the case
that one of the distributions is explicitly given~\cite{DKN15a}. Table~1 summarizes our upper bounds for 
a selection of natural and well-studied distribution families. These results are obtained from Theorem~\ref{thm:main} 
and Fact~\ref{fact:simple}, via the appropriate structural approximation results~\cite{CDSS13, CDSS14}.

\begin{table*}[t]
\begin{center}
\begin{tabular}{|c|c|c|c|}%
\hline \bf Distribution Family & \bf Our upper bound & \bf Previous work
\\\hline\hline

$t$-piecewise constant & $O\big(\max \big\{ \frac{t^{4/5}}{\eps^{6/5}}, \frac{t^{1/2}}{\eps^2} \big\}\big)$ & $O\big(\frac{t}{\eps^2}\big)$ \cite{CDSS14} \\ \hline

$t$-piecewise degree-$d$ & $O\big(\max \big\{ \frac{(t(d+1))^{4/5}}{\eps^{6/5}}, \frac{(t(d+1))^{1/2}}{\eps^2} \big\}\big)$ &  $O\big(\frac{t(d+1)}{\eps^2}\big)$  \cite{CDSS14}
\\\hline

log-concave & $O\big(\frac{1}{\eps^{9/4}}\big)$ & $O\big(\frac{1}{\eps^{5/2}}\big)$ \cite{CDSS14} \\ \hline

$k$-mixture of log-concave &$O\big(\max \big\{ \frac{k^{4/5}}{\eps^{8/5}}, \frac{k^{1/2}}{\eps^{9/4}} \big\}\big)$& $O\big(\frac{k}{\eps^{5/2}}\big)$ \cite{CDSS14}  \\ \hline

$t$-modal over $[n]$ &$O\big(\max \big\{ \frac{(t \log n)^{4/5}}{\eps^{2}}, \frac{(t \log n)^{1/2}}{\eps^{5/2}} \big\}\big)$&  $O\big(\frac{(t \log n)^{2/3}}{\eps^{{8}/3}}+ \frac{t^2}{\eps^4}\big)$ \cite{DDSVV13}   \\\hline


MHR over $[n]$ &$O(\max \big\{ \frac{\log(n/\eps)^{4/5}}{\eps^{2}}, \frac{\log(n/\eps)^{1/2}}{\eps^{5/2}} \big\})$ &  $O\big(\frac{\log(n/\eps)}{\eps^{3}}\big)$  \cite{CDSS14} \\\hline

\end{tabular}
\end{center}
\label{table:results}
\caption{Algorithmic results for closeness testing of selected families of structured probability distributions. 
The second column indicates the sample complexity of our general algorithm applied to the class under consideration.
The third column indicates the sample complexity of the best previously known algorithm for the same problem.
}
\label{tab:results}
\end{table*}

We would like to stress that our algorithm and its analysis are
very different than previous results in the property testing literature.
We elaborate on this point in the following subsection.

\subsection{Our Techniques} \label{ssec:techniques}
In this subsection, we provide a high-level overview of our techniques in tandem with a comparison to prior work.

Our upper bound is achieved by an explicit, sample near-linear-time algorithm. 
A good starting point for considering this problem would be the testing algorithm of~\cite{DKN15a},
which deals with the case where $p$ is an {\em explicitly known} distribution. The basic idea of the testing algorithm in this case~\cite{DKN15a}
is to partition the domain into intervals in several different ways, and run a known $L_2$ tester on the reduced distributions 
(with respect to the intervals in the partition) as a black-box. 
At a high-level, these intervals partitions can be constructed by exploiting 
our knowledge of $p$, in order to divide our domain into several 
equal mass intervals under $p$. 
It can be shown that if $p$ and $q$ have large $\mathcal{A}_k$ distance from each other,
one of these partitions will be able to detect the difference.

Generalizing this algorithm to the case where $p$ is {\em unknown} turns out to be 
challenging, because there seems to be no way to find the appropriate interval partitions with $o(k)$ samples.
If we allowed ourselves to take $\Omega(k/\epsilon)$ samples from $p$, 
we would be able to approximate an appropriate interval partition, 
and make the aforementioned approach go through.
Alas, this would not lead to an $o(k)$ sample algorithm.
If we can only draw $m$ samples from our distributions, the best that we could hope to do
would be to use our samples in order to partition the domain into $m+1$ interval regions.
This, of course, is not going to be sufficient to allow an analysis 
along the lines of the above approach to work. 
In particular, if we partition our domain \emph{deterministically} into $m=o(k)$ intervals, 
it may well be the case that the reduced distributions over those intervals 
are identical, 
despite the fact that the original distributions have large $\mathcal{A}_k$ distance. 
In essence, the differences between $p$ and $q$ 
may well cancel each other out on the chosen intervals. 

However, it is important to note that our interval boundaries are \emph{not} deterministic. 
This suggests that unless we get unlucky, the discrepancy between 
$p$ and $q$ will not actually cancel out in our partition. 
As a slight modification of this idea, instead of
partitioning the domain into intervals (which we expect to have only $O(1)$ samples each) 
and comparing the number of samples from $p$ versus $q$ in each,
we sort our samples and test how many of them came 
from the same distribution as their neighbors (with respect to the natural ordering on the real line).

We intuitively expect that, if $p = q$, the number of pairs of ordered samples drawn from the same distribution
versus a different one will be the same. Indeed, this can be formalized and the completeness
of this tester is simple to establish. The soundness analysis, however, is somewhat involved. 
We need to show that the expected value of the statistic that we compute 
is larger than its standard deviation. While the variance is easy to bound from above, bounding the expectation is quite challenging. 
To do so, we define a function, $f(t)$, 
that encodes how likely it is that the samples nearby point $t$ 
come from one distribution or the other. 
It turns out that $f$ satisfies a relatively nice differential equation, 
and relates in a clean way to the expectation of our statistic. 
From this, we can show that any discrepancy between $p$ and $q$ 
taking place on a scale too short to be detected by {the above partitioning approach} 
will yield a notable contribution to our expectation.

The analysis of our lower bound begins by considering a natural class of testers, 
namely those that take some number of samples from $p$ and $q$, 
sort the samples (while keeping track of which distribution they came from) 
and return an output that depends only on the ordering of these samples. 
For such testers we exhibit explicit families of pairs of distributions that are 
hard to distinguish from being identical. There is a particular pattern that appears 
many times in these examples, where there is a small interval for which $q$ has 
an appropriate amount of probability mass, followed by an interval of $p$, followed 
by another interval of $q$. When the parameters are balanced correctly, it can be 
shown that when at most two samples are drawn from this subinterval, 
the distribution on their orderings is indistinguishable from the case where $p=q$. 
By constructing distributions with many copies of the pattern, we essentially show 
that a tester of this form will not be able to be confident that $p\neq q$, 
unless there are many of these small intervals from which it draws three or more samples. 
On the other hand, a simple argument shows that this is unlikely to be the case.

The above lower bound provides explicit distributions that are hard to distinguish 
from being identical by any tester in this limited class. To prove a lower bound against 
general testers, we proceed via a reduction: we show that an order--based tester 
can be derived from any general tester. It should be noted that this makes our lower bound 
in a sense non-constructive, as we do not know of any explicit families of distributions 
that are hard to distinguish from uniform for general testers.  In order to perform this reduction, 
we show that for a general tester we can find some large subset $S$ of its domain 
such that if all samples drawn from $p$ and $q$ by the tester happen to lie in $S$, 
then the output of the tester will depend only on the ordering of the samples. 
This essentially amounts to a standard result from Ramsey theory. 
Then, by taking any other problem, we can embed it into our new sample space 
by choosing new $p$ and $q$ that are the same up to an order-preserving rearrangement 
of the domain (which will also preserve $\mathcal{A}_k$ distance), ensuring that they are supported only on $S$.

\section{Algorithm for $\mathcal{A}_k$ Closeness Testing} \label{sec:alg}

In this section we provide the sample optimal closeness tester under the $\Ak$ distance.

\subsection{An $O(k^{4/5}/\eps^{6/5})$-sample tester} \label{ssec:simple-alg}

In this subsection we give a tester with sample complexity $O(k^{4/5}/\eps^{6/5})$
that applies for $\eps = \Omega(k^{-1/6})$.
For simplicity, we focus on the case that we take samples
from two unknown distributions with probability density functions $p, q: [0, 1] \to \R_+$.
Our results are easily seen to extend to discrete probability distributions.

\medskip

\fbox{\parbox{6in}{
{\bf Algorithm} Simple-Test-Identity-$\mathcal{A}_k(p, q, \eps)$\\
Input: sample access to pdf's $p, q: [0, 1] \to \R_+$, $k \in \Z_+$, and $\eps > 0$.\\
Output: ``YES'' if $q = p$; ``NO'' if $\|q-p\|_{\mathcal{A}_k} \ge \eps.$


\begin{enumerate}

\item Let $m = C\cdot(k^{4/5}/\eps^{6/5})$, for a sufficiently large constant $C$.
Draw two sets of samples $S_p$, $S_q$ each of size $\Poi(m)$ from $p$ and from $q$ respectively.


\item Merge $S_p$ and $S_q$ while remembering from which distribution each sample comes from.
Let $S$ be the union of $S_p$ and $S_q$ sorted in increasing order (breaking ties randomly).


\item Compute the statistic $Z$ defined as follows:
\begin{eqnarray*}
Z \eqdef &\#&(\mbox{pairs of successive samples in $S$ coming from the same distribution}) -\\
&\#&(\mbox{pairs of successive samples in $S$ coming from different distributions})
\end{eqnarray*}


\item If $Z > 3\cdot(\sqrt{m})$ return "NO". Otherwise return "YES".

\end{enumerate}
}}

\begin{proposition} \label{prop:ak-simple}
The algorithm {\em Simple-Test-Identity-}$\mathcal{A}_k(p, q, \eps)$, on input two samples each of size $O(k^{4/5}/\eps^{6/5})$
drawn from two distributions with densities $p, q: [0,1] \to \R_+$, an integer $k > 2$,
and $\eps = \Omega (k^{-1/6})$, correctly distinguishes the case that $q=p$ from the case
$\lVert p-q \rVert_{\mathcal{A}_k} \ge \eps$, with probability at least $2/3$.
\end{proposition}

\begin{proof}
First, it is straightforward to verify the claimed sample complexity, since the algorithm only draws samples in Step~1.
To simplify the analysis we make essential use of the following simple claim:
\begin{claim}
We can assume without loss of generality that the pdf's $p, q : [0, 1] \to \R_+$ are continuous functions bounded from above by $2$.
\end{claim}
\begin{proof}
We start by showing we can assume that $p, q$ are at most $2$.
Let $p, q : [0, 1] \to \R_+$ be arbitrary pdf's.
We consider the cumulative distribution function (CDF)  $\Phi$ of the mixture $(p+q)/2$.
Let $X \sim p$, $Y \sim q$, $W \sim (p+q)/2$ be random variables.
Since $\Phi$ is non-decreasing, replacing $X$ and $Y$ by $\Phi(X)$ and $\Phi(Y)$
does not affect the algorithm (as the ordering on the samples remains the same).
We claim that, after making this replacement, $\Phi(X)$ and $\Phi(Y)$ are continuous distributions
with probability density functions bounded by $2$.
In fact, we will show that the sum of their probability density functions is exactly $2$.
This is because for any $0\leqslant a \leqslant b\leqslant 1$,
$$
\Pr[\Phi(X)\in [a,b]] + \Pr[\Phi(Y)\in[a,b]]=2 \Pr[\Phi(W) \in [a,b]] = 2(b-a) \;,
$$
where the second equality is by the definition of a CDF. Thus, we can assume that $p$ and $q$ are bounded from above by $2$.

To show that we can assume continuity, note that $p$ and $q$ can be approximated by continuous density functions
$p'$ and $q'$ so that the $L_1$ errors $\|p-p'\|_1,\|q-q'\|_1$ are each at most $1/(10m)$.
If our algorithm succeeds with the continuous densities $p'$ and $q'$, it must also succeed for $p$ and $q$.
Indeed, since the $L_1$ distance between $p$ and $p'$ and $q$ and $q'$ is at most $1/(10m)$,  a set of
$m$ samples taken from $p$ or $q$ are statistically indistinguishable to $m$ samples taken from $p'$ or $q'$.
This proves that it is no loss of generality to assume that $p$ and $q$ are continuous.
\end{proof}

Note that the algorithm makes use of the well-known ``Poissonization'' approach.
Namely, instead of drawing $m=O(k^{4/5}/\eps^{6/5})$ samples from $p$ and from $q$,
we draw $m'=\Poi(m)$ samples from $p$ and $m''=\Poi(m)$ sample from $q$.
The crucial properties of the Poisson distribution are that it is sharply concentrated around its mean
and it makes the number of times different elements occur in the sample independent.


We now establish completeness. Note that our algorithm draws $\Poi(2m)$ samples from $p$ or $q$.
If $p=q$, then our process equivalently selects $\Poi(2m)$ values from $p$ and then randomly and independently
with equal probability decides whether or not each sample came from $p$ or from $q$.
Making these decisions one at a time in increasing order of points, we note that each adjacent pair of elements in
$S$ randomly and independently contributes either a $+1$ or a $-1$ to $Z$.
Therefore, the distribution of $Z$ is exactly that of a sum of $\Poi(2m)-1$ independent $\{\pm 1\}$ random variables.
Therefore, $Z$ has mean $0$ and variance $2m-1$. By Chebyshev's inequality it follows that $|Z|\leqslant 3\sqrt{m}$ with probability at least $7/9$.
This proves completeness.

We now proceed to prove the soundness of our algorithm.
Assuming that $\lVert p-q \rVert_{\mathcal{A}_k} > \eps$,
we want to show that the value of $Z$ is at most $3 \cdot \sqrt{m}$ with probability at most $1/3$.
To prove this statement, we will again use Chebyshev's inequality.
In this case it suffices to show that $\E[Z] \gg \sqrt{\Var[Z]}+\sqrt{m}$ for the inequality to be applicable.
We begin with an important definition.
\begin{definition}
Let $f: [0,1] \to [-1, 1]$ equal
\begin{align*}
f(t) \eqdef  \ \ \ &
\Pr\left[\textrm{largest sample in }S\textrm{ that is at most }t\textrm{ was drawn from }p \right] \\ -
&\Pr\left[\textrm{largest sample in }S\textrm{ that is at most }t\textrm{ was drawn from }q \right] \; .
\end{align*}
\end{definition}
\noindent The importance of this function is demonstrated by the following lemma.
\begin{lemma}\label{EZLem}
We have that: $\E[Z] = m\int_0^1 f(t)(p(t) - q(t))dt \;.$
\end{lemma}
\begin{proof}
{Given an interval $I$, we let $Z_I$ be the contribution to $Z$ coming from pairs of consecutive points of $S$ 
the larger of which is drawn from $I$. We wish to approximate the expectation of $Z_I$. 
We let $\tau(I) = m(p(I)+q(I))$ be the expected total number of points drawn from $I$. 
We note that the contribution coming from cases where more than one point is drawn from $I$ is $O(\tau(I)^2)$. 
We next consider the contribution under the condition that only one sample is drawn from $I$. For this, 
we let $\mathrm{EP}_I$ and $\mathrm{EQ}_I$ be the events that the largest element of $S$ preceding $I$ comes from $p$ or $q$ respectively. 
We have that the expected contribution to $Z_I$ coming from events where exactly one element of $S$ is drawn from $I$ is:
\begin{align*}
& (\Pr[\mathrm{EP}_I]-\Pr[\mathrm{QP}_I])\Pr(\textrm{The only element drawn from }I\textrm{ is from }p) \\ - & (\Pr[\mathrm{EP}_I]-\Pr[\mathrm{QP}_I])\Pr(\textrm{The only element drawn from }I\textrm{ is from }q).
\end{align*}
Letting $x_I$ be the left endpoint of $I$, this is
$$
f(x_I)(mp(I)-mq(I))+O(\tau(I)^2).
$$
Therefore,
$$
\E[Z_I] = f(x_I)(mp(I)-mq(I))+O(\tau(I)^2).
$$
Letting $\mathcal{I}$ be a partition of our domain into intervals, we find that
\begin{align*}
\E[Z] & = \sum_{I\in\mathcal{I}} \E[Z_I]\\
& = \sum_{I\in\mathcal{I}} f(x_I)(mp(I)-mq(I))+O(\tau(I)^2)\\
& = O(m \max_{I\in\mathcal{I}} \tau(I)) + \sum_{I\in\mathcal{I}} f(x_I)(mp(I)-mq(I)).
\end{align*}
As the partition $\mathcal{I}$ becomes iteratively more refined, these sums approach Riemann sums for the integral of
$$
mf(x)(p(x)-q(x))dx.
$$
Therefore, taking a limit over partitions $\mathcal{I}$, we have that
$$
\E[Z] = m\int f(x)(p(x)-q(x)) dx.
$$

}
\end{proof}
We will also make essential use of the following technical lemma:
\begin{lemma} \label{lem:diff}
The function $f$ is differentiable with derivative
$ f'(t) = m\left( p(t) - q(t) - (p(t)+q(t))f(t) \right).$
\end{lemma}
\begin{proof}
Consider the difference between $f(t)$ and $f(t+h)$ for some small $h>0$.
We note that $f(t)=\E[F_t]$ where $F_t$ is $1$ if the sample of $S$ preceding $t$ came from $p$,
$-1$ if the sample came from $q$, and $0$ if no sample came before $t$. Note that
$$
F_{t+h} = \begin{cases} F_t & \textrm{ if no samples from }p\textrm{ nor }q\textrm{ are drawn from }[t,t+h]\\
1 & \textrm{ if one sample from }p\textrm{ and none from }q\textrm{ are drawn from }[t,t+h]\\
-1 & \textrm{ if one sample from }q\textrm{ and none from }p\textrm{ are drawn from }[t,t+h]\\
\pm1 & \textrm{ if at least two samples from }p\textrm{ or }q \textrm{ are drawn from }[t,t+h].\end{cases}
$$
Since $p$ and $q$ are continuous at $t \in [0, 1]$, these four events happen with probabilities $1-mh(p(t)+q(t))+o(h)$, $mhp(t)+o(h)$, $mhq(t)+o(h)$, $o(h)$, respectively. Therefore, taking an expectation we find that $f(t+h)=f(t)(1-mh(p(t)+q(t)))+mh(p(t)-q(t))+o(h)$. This, and a similar relation relating $f(t)$ to $f(t-h)$, proves that $f$ is differentiable with the desired derivative.
\end{proof}
To analyze the desired expectation, $\E[Z]$, we consider the quantity
$\int_0^1 f'(t)f(t)dt = (1/2) \left( f^2(1) - f^2(0)\right)$.
Substituting $f'$ from Lemma~\ref{lem:diff} above
gives $$\int_0^1 f'(t)f(t)dt = m\int_0^1f(t)(p(t)-q(t))dt\\ - m\int_0^1f^2(t)(p(t)+q(t))dt.$$
Combining this with Lemma \ref{EZLem},
we get
\begin{equation}\label{eq:exp}
\E[Z] = m\int_0^1 f^2(t)(p(t) + q(t))dt + f^2(1)/2 \; .
\end{equation}
The second term in \eqref{eq:exp} above is $O(1)$, so we focus our attention to bound the first term from below.
To do this, we consider intervals $I \subset [0,1]$
over which $|p(I)-q(I)|$ is ``large'' and show that they must produce some noticeable contribution to the first term.
Fix such an interval $I$.
We want to show that $f^2$ is large somewhere in $I$.
Intuitively, we attempt to prove that on at least one of the endpoints of the interval,
the value of $f$ is big. Since $f$ does not vary too rapidly, $f^2$ will be large on some large fraction of $I$.
Formally, we have the following lemma:


\begin{lemma}\label{largeEndLem}
For $\delta >0$,
let $I \subset [0, 1]$ be an interval with $| p(I) - q(I) | = \delta$ and $p(I)+q(I) < 1/m$.
Then, there exists an $x \in I $ such that $|f(x)| \geqslant \frac{m\delta}{3}.$
\end{lemma}
\begin{proof}
Suppose for the sake of contradiction that $|f(x)|<  m\delta/3$ for all $x\in I = [X,Y]$.
Then, we have that
\begin{eqnarray*}
2m\delta/3 &>& |f(X)-f(Y)| = \left|\int_X^Y f'(t)dt\right| = \left|\int_X^Y \left( m(p(t)-q(t)) -mf(t)(p(t)+q(t)) \right) dt\right| \\
&=& \left|m(p(I)-q(I)) -m\int_X^Y f(t)(p(t)+q(t))dt \right|
\geqslant m|p(I)-q(I)| -m \left|\int_X^Yf(t)(p(t)+q(t))dt \right|\\
&>& m\delta -m\int_X^Y \left(m\delta/3\right)(p(t)+q(t))dt = m \delta\left(1-m(p(I)+q(I))/3 \right) > 2m\delta/3 \;,
\end{eqnarray*}
which yields the desired contradiction.
\end{proof}






\noindent We are now able to show that the contribution to $\E[Z]$ coming from such an interval is large.
\begin{lemma}\label{ZContLem}
Let $I$ be an interval satisfying the hypotheses of Lemma \ref{largeEndLem}.
Then
$$\int_I f^2(t)(p(t)+q(t))dt  =  \Omega(m^2\delta^3) \;.$$
\end{lemma}
\begin{proof}
By Lemma \ref{largeEndLem},  $f$ is large at some point $x$ of the interval $I = [X, Y]$.
Without loss of generality, we assume that $p([X,x])+q([X,x]) \leqslant (p(I)+q(I))/2$.
Let $I'=[x,Y']$ be the interval so that $p(I')+q(I')=\delta/9$.
Note that $I'\subset I$ {(since by assumption $|p(I)-q(I)|>\delta$ and thus $p(I)+q(I)>\delta$)}. Furthermore, note that since with probability at least $1-m\delta/9$,
no samples from $S$ lie in $I'$, we have that for all $z$ in $I'$ it holds
$|f(x)-f(z)|\leqslant 2m\delta/9$, so $|f(z)|\geqslant m\delta/9$. Therefore,
\begin{eqnarray*}
\int_I f^2(t)(p(t)+q(t))dt &\geqslant& \int_{I'}f^2(t)(p(t)+q(t))dt \geqslant \int_{I'} \left(\frac{m\delta}{9}\right)^2 (p(t)+q(t))dt\\
&=& \frac{m^2\delta^2}{81}(p(I')+q(I')) = \frac{m^2\delta^3}{729} \;.
\end{eqnarray*}
\end{proof}
Since $\| p-q\|_{\mathcal{A}_k} >\eps$, there is a partition $\mathcal{I}$ of $[0,1]$
into $k$ intervals so that $\| p_r^{\mathcal{I}} - q_r^{\mathcal{I}}\|_1 >\eps.$
By subdividing intervals further if necessary, we can guarantee that $\mathcal{I}$ has at most $3k$ intervals,
$\|p_r^{\mathcal{I}} - q_r^{\mathcal{I}}\| >\eps,$ and for each subinterval $I \in \mathcal{I}$ it holds $p(I), q(I)\leqslant 1/k$.
For each such interval $I \in \mathcal{I}$, let $\delta_I=|p(I)-q(I)|$.
Note that $\sum_{I\in\mathcal{I}} \delta_I\geqslant \eps$.

By \eqref{eq:exp} we have that
\begin{eqnarray*}
\E[Z] &=&  m \sum_{I \in \mathcal{I}} \int_I f^2(t)(p(t)+q(t))dt +O(1)\\
&=& \Omega \left( m\sum_{I \in \mathcal{I}} m^2 \delta_I^3 \right)
= \Omega \left( m^3 (\sum_{I \in \mathcal{I}}  \delta_I)^3/(3k)^2 \right) \\
& = & \Omega \left(m^3\eps^3/k^2 \right) = \Omega (C^{5/2}\sqrt{m}) \;.
\end{eqnarray*}
{We note that the second to last line above follows by H\"{o}lder's inequality.}
It remains to bound from above the variance of $Z$.
\begin{lemma} \label{lemma:variance}
We have that
$\Var[Z]  = O(m) \;.$
\end{lemma}
\begin{proof}
We divide the domain
$[0,1]$ into $m$ intervals $I_i$, $i=1, \ldots, m$,
each of total mass $2/m$
under the sum-distribution $p+q$.
Consider the random variable $X_i$ denoting
the contribution to $Z$ coming from pairs of adjacent samples in $S$
such that the right sample is drawn from  $I_i$.
Clearly, $Z = \sum_{i=1}^m X_i$ and
$\Var[Z] = \sum_{i=1}^m \Var[X_i] + \sum_{i\neq j} \mathrm{Cov}(X_i, X_j).$

To bound the first sum, note that the number of pairs of $S$ in an interval $I_i$
is no more than the number of samples drawn from $I_i$, and the variance of $X_i$
is less than the expectation of the square of the number of samples from $I_i$.
Since the number of samples from $I_i$ is a Poisson random variables with parameter $2$, we have
$\Var[X_i] =O(1)$. This shows that $ \sum_{i=1}^m \Var[X_i] = O(m) .$

To bound the sum of covariance,
consider $X_i$ and $X_j$ conditioned on the samples drawn from intervals other than $I_i$ and $I_j$. Note that if any sample is drawn from an intermediate interval, $X_i$ and $X_j$ are uncorrelated, and otherwise their covariance is at most $\sqrt{\var(X_i)\var(X_j)}=O(1)$.
Since the probability that no sample is drawn from any intervening interval decreases exponentially with their separation,
it follows that $\mathrm{Cov}(X_i, X_j) = O(1)\cdot e^{-\Omega(|j-i|)}$. This completes the proof.
\end{proof}
An application of Chebyshev's inequality completes the analysis of the soundness and the proof of Proposition~\ref{prop:ak-simple}.
\end{proof}

\subsection{The General Tester} \label{ssec:gen-ub}
In this section, we  present a tester whose sample complexity is optimal (up to constant factors)
for all values of $\eps$ and $k$, thereby establishing the upper bound part of Theorem~\ref{thm:main}.
Our general tester (Algorithm Test-Identity-$\mathcal{A}_k$) builds on the tester presented in the previous subsection (Algorithm  Simple-Test-Identity-$\mathcal{A}_k$).
It is not difficult to see that the latter algorithm can fail once $\eps$ becomes sufficiently small, 
if the discrepancy between $p$ and $q$ is concentrated on intervals of
mass larger than $1/m$. In this scenario, the tester Simple-Test-Identity-$\mathcal{A}_k$
will not take sufficient advantage of these intervals. 
To obtain our enhanced tester Test-Identity-$\mathcal{A}_k$, we will need to combine
Simple-Test-Identity-$\mathcal{A}_k$ with an alternative tester when this is the case.
Note that we can easily bin the distributions $p$ and $q$ into intervals of total mass
approximately $1/m$ by taking $m$ random samples.
Once we do this, we can use an identity tester similar to that in our previous work~\cite{DKN15a} to detect the discrepancy in these intervals.
In particular we show the following:

\begin{proposition}\label{UniformTesterProp}
Let $p, q$ be discrete distributions over $[n]$ satisfying $\| p\|_2,\|q\|_2 = O(1/\sqrt{n})$.
There exists a testing algorithm with the following properties: On input
$k \in \Z_+$, $ 2 \le k \le n$, and $\delta, \eps>0$, the algorithm draws $O\left((\sqrt{k}/\eps^2) \cdot \log(1/\delta)\right)$ samples
from $p$ and $q$ and with probability at least $1-\delta$ distinguishes between the cases $p=q$ and $\|p-q\|_{\mathcal{A}_k} > \eps$.
 \end{proposition}

The above proposition says that the identity testing problem under the $\mathcal{A}_k$ distance can be solved with
$O(\sqrt{k}/\eps^2)$ samples when both distributions $p$ and $q$ are promised to be ``nearly'' uniform (in the sense that their
$L_2$ norm is $O(1)$ times that of the uniform distribution). To prove Proposition~\ref{UniformTesterProp} we follow a similar approach
as in~\cite{DKN15a}: Starting from the $L_2$ identity tester of~\cite{CDVV14},
we consider several oblivious interval decompositions of the domain into intervals of approximately the same mass, and apply
a ``reduced'' identity tester for each such decomposition.
The details of the analysis establishing Proposition~\ref{UniformTesterProp} are postponed to Appendix~\ref{sec:app}.


\medskip

\noindent We are now ready to present our general testing algorithm:

\medskip

 \fbox{\parbox{6in}{
{\bf Algorithm} Test-Identity-$\mathcal{A}_k(p, q, \eps)$\\
Input: sample access to distributions $p, q: [0, 1] \to \R_+$, $k \in \Z_+$, and $\eps>0$.\\
Output: ``YES'' if $q = p$; ``NO'' if $\|q-p\|_{\mathcal{A}_k} \ge \eps.$
\begin{enumerate}
\item Let $m=C k^{4/5}/\eps^{6/5}$, for a sufficiently large constant $C$.
Draw two sets of samples $S_p$, $S_q$ each of size $\Poi(m)$ from $p$ and from $q$ respectively.


\item Merge $S_p$ and $S_q$ while remembering from which distribution each sample comes from.
Let $S$ be the union of $S_p$ and $S_q$ sorted in increasing order (breaking ties randomly).


\item Compute the statistic $Z$ defined as follows:
\begin{eqnarray*}
Z \eqdef &\#&(\mbox{pairs of successive samples in $S$ coming from the same distribution}) -\\
&\#&(\mbox{pairs of successive samples in $S$ coming from different distributions})
\end{eqnarray*}


\item  If $Z>5\sqrt{m}$ return ``NO''.


\item Repeat the following steps $O(C)$ times:
\begin{itemize}
\item[(a)] Draw $\Poi(m)$ samples from $(p+q)/2$.
\item[(b)] Split the domain into intervals with the interval endpoints given by the above samples.
Let $p'$ and $q'$ be the reduced distributions with respect to these intervals.

\item[(c)]  Run the tester of Proposition \ref{UniformTesterProp} on $p'$ and $q'$ with error probability $1/C^2$
to determine if $\|p'-q'\|_{\mathcal{A}_{2k+1}}>\eps/C$. If the output of this tester is ``NO'', output ``NO''.
\end{itemize}


\item Output ``YES''.
\end{enumerate}

}}

\medskip

Our main result for this section is the following:

\begin{theorem} \label{thm:main-upper}
Algorithm Test-Identity-$\mathcal{A}_k$
draws $O(\max\{k^{4/5}/\eps^{6/5},k^{1/2}/\eps^2\})$ samples from $p, q$
and with probability at least $2/3$ returns ``YES'' if $p=q$ and ``NO'' if $\|p-q\|_{\mathcal{A}_k} > \eps$.
\end{theorem}
\begin{proof}
First, it is easy to see that the sample complexity of the algorithm is $O(m+k^{1/2}/\eps^2)$.
Recall that we can assume that $p, q$ are continuous pdf's bounded from above by $2$.

We start by establishing completeness. If $p=q$, it is once again the case that $\E[Z]=0$ and $\var[Z]<2m$,
so by Chebyshev' s inequality, Step~4 will fail with probability at most $1/9$.
Next when taking our samples in Step~5(a), note that the expected  samples size is $O(m)$
and that the expected squared $L_2$ norms of the reduced distributions $p'$ and $q'$ are $O(1/m)$.
Therefore, with probability at least $1-1/C^2$,
$p'$ and $q'$ satisfy the hypothesis of Proposition \ref{UniformTesterProp}.
Hence,  this holds for all $C$ iterations with probability at least $8/9$.

Conditioning on this event, since $p'=q'$,
the tester in Step~5(c) will return ``YES'' with probability at least $1-1/C^2$ on each iteration.
Therefore, it returns ``YES'' on all iterations with probability at least $8/9$.
By a union bound, it follows that  if $p=q$, our algorithm returns ``YES'' with probability at least $2/3$.

We now proceed to establish soundness.
Suppose that $\|p-q\|_{\mathcal{A}_k} \geqslant \eps$.
Then there exists a partition $\mathcal{I}$ of the domain into $k$ intervals such that
$\|p_r^{\mathcal{I}}-q_r^{\mathcal{I}}\| \geqslant \eps.$
For an interval $I\in\mathcal{I}$, let $\delta(I) = |p(I)-q(I)|$. We will call an $I\in\mathcal{I}$ small
if there is a subinterval $J \subseteq I$ so that $p(J)+q(J)<1/m$ and $|p(J)-q(J)|\geqslant \delta(I)/3$.
We will call $I$ large otherwise.
Note that $
\sum_{I\in\mathcal{I},I  \textrm{ small}} \delta(I) +
\sum_{I\in\mathcal{I},I  \textrm{ large}} \delta(I) =
\sum_{I\in\mathcal{I}} \delta(I) \geqslant \eps.$
Therefore either $ \sum_{I\in\mathcal{I},I  \textrm{ small}} \delta(I) \geqslant \eps/2,$
or $ \sum_{I\in\mathcal{I},I \textrm{ large}} \delta(I) \geqslant \eps/2.$
We analyze soundness separately in each of these cases.

Consider first the case that $\sum_{I\in\mathcal{I},I  \textrm{ small}} \delta(I) \geqslant \eps/2.$
The analysis in this case is very similar to the soundness proof of Proposition~\ref{prop:ak-simple}
which we describe for the sake of completeness.

By definition,
for each small interval $I$, there exists a subinterval $J$
so that $p(J)+q(J)<1/m$ and $|p(J)-q(J)|>\delta(I)/2$.
By Lemma \ref{ZContLem}, for such $J$ we have that
$ \int_J f^2(t)(p(t)+q(t))dt = \Omega(m^2 \delta^3(I)),$
and therefore, that $\int_I f^2(t)(p(t)+q(t))dt = \Omega(m^2 \delta^3(I)).$
Hence, we have that
\begin{align*}
\E[Z] & \geqslant m\int_{0}^1 f^2(t)(p(t)+q(t))dt\\
& \geqslant \sum_{I\in\mathcal{I},I  \textrm{ small}} m\int_I f^2(t)(p(t)+q(t))dt\\
& \geqslant \sum_{I\in\mathcal{I},I  \textrm{ small}} \Omega(m^3\delta^3(I))\\
& \geqslant \Omega(m^3) \left( \sum_{I\in\mathcal{I},I  \textrm{ small}} \delta(I) \right)^3/k^2\\
& = \Omega(m^3\epsilon^3/k^2)\\
& = \Omega(C^{5/2}\sqrt{m}).
\end{align*}
On the other hand, Lemma~\ref{lemma:variance} gives that $\var[Z] = O(m)$, so for $C$ sufficiently large,
Chebyshev's inequality implies that with probability at least $2/3$ it holds $Z>5\sqrt{m}$.
That is, our algorithm outputs ``NO'' with probability at least $2/3$.

Now consider that case that $\sum_{I\in\mathcal{I},I  \textrm{ large}} \delta(I) \geqslant \eps/2.$
We claim that the second part of our tester will detect the discrepancy between $p$ and $q$ with high constant probability.
Once again, we can say that with probability at least $8/9$ the squared $L_2$ norms of the reduced distributions $p'$ and $q'$ are both $O(1/m)$
and that the size of the reduced domain is $O(m)$. Thus, the conditions of Proposition \ref{UniformTesterProp} are satisfied on all iterations with probability at least $8/9$.
To complete the proof, we will show that with constant probability we have $\|p'-q'\|_{\mathcal{A}_{2k+1}}>\eps/C$.
To do this, we construct an explicit partition $\mathcal{I'}$ of our reduced domain into at most $2k+1$ intervals so that with constant probability
$\|p_r^{\mathcal{I'}}-q_r^{\mathcal{I'}}\|_1 >\eps/C$. This will imply that with probability at least $8/9$ that on at least one of our $C$ trials that $\|p_r^{\mathcal{I'}}-q_r^{\mathcal{I'}}\|_1 >\eps/C$.

More specifically, for each interval $I\in\mathcal{I}$ we place interval boundaries at the smallest and largest sample points taken from $I$ in Step~5(a)
(ignoring them if fewer than two samples landed in $I$). Since we have selected at most $2k$ points, this process
defines a partition  $\mathcal{I'}$ of the domain into at most $2k+1$ intervals.
We will show that the reduced distributions $p''=p_r^{\mathcal{I'}}$ and $q''=q_r^{\mathcal{I'}}$ have large expected $L_1$ error.

In particular, for each interval $I\in\mathcal{I}$ let $I'$ be the interval between the first and last sample points of $I$. Note that $I'$ is an interval in the partition $\mathcal{I'}$. We claim that if $I$ is large, then with constant probability
$$
|p(I')-q(I')| = \Omega(\delta(I)).
$$
Let $I=[X,Y]$ and $I'=[x,y]$ (so $x$ and $y$ are the smallest and largest samples taken from $I$, respectively). We note that if $p([X,x])+q([X,x])<1/m$ and $p([y,Y])+q([y,Y])<1/m$ then
$$
|p(I')-q(I')| \geqslant |p(I)-q(I)| - |p([X,x])-q([X,x])| - |p([y,Y])-q([y,Y])| \geqslant \delta(I) - \delta(I)/3- \delta(I)/3 = \delta(I)/3,
$$
where the second inequality uses the fact that $I$ is large. On the other hand, we note that $p([X,x])+q([X,x])$ and $p([y,Y])+q([y,Y])$ are exponential distributions with mean $1/m$, and thus, this event happens with constant probability. Let $N_I$ be the indicator random variable for the event that $|p(I')-q(I')|\geqslant \delta(I)/3$. We have that
$$
\|p''-q''\|_1 \geqslant \sum_I N_I\delta(I)/3 \geqslant \sum_{I\in\mathcal{I},I \textrm{ large}} N_I\delta(I)/3.
$$
Thus, we have that
$$
\|p''-q''\|_1 \geqslant \sum_{I\in\mathcal{I},I  \textrm{ large}} \delta(I)/3 - \sum_{I\in\mathcal{I},I \textrm{ large}}(1-N_I)\delta(I)/3.
$$
Therefore, since
$$
\E\left[  \sum_{I\in\mathcal{I},I  \textrm{ large}}(1-N_I)\delta(I)/3\right] <
\left(\sum_{I\in\mathcal{I},I  \textrm{ large}}\delta(I)/3\right)(1-c)
$$
for some fixed $c>0$, we have that with constant probability that
$$
\|p''-q''\|_1 \geqslant c \sum_{I\in\mathcal{I},I \textrm{ large}} \delta(I)/3 \geqslant c\eps/6 \geqslant \eps/C.
$$
This means that with probability at least $8/9$ for at least one iteration
we will have that $\|p'-q'\|_{\mathcal{A}_{2k+1}}>\eps/C$, and therefore, with probability at least $2/3$, our algorithm outputs ``NO''.

\end{proof}

\section{Lower Bound for $\mathcal{A}_k$ Closeness Testing} \label{sec:lb}

Our upper bound from Section \ref{sec:alg} seems potentially suboptimal.
Instead of obtaining an upper bound of $O( \max\{ k^{2/3}/\eps^{4/3}, k^{1/2}/\eps^2 \})$,
which would be analogical to the unstructured testing result of~\cite{CDVV14},
we obtain a very different bound of $O( \max\{ k^{4/5}/\eps^{6/5}, k^{1/2}/\eps^2 \})$.
In this section we show, surprisingly, that our upper bound is optimal for 
continuous distributions, or discrete distributions with support size $n$ that is sufficiently large
as a function of $k$.


Intuitively, our lower bound proof consists of two steps.
In the first step, we show it is no loss of generality to assume 
that an optimal algorithm only considers the ordering of the samples,
and ignores all other information. In the second step, we construct a pair of distributions
which is hard to distinguish given the condition that the tester is only allowed
to look at the ordering of the samples and nothing more.

Our first step is described in the following theorem. 
{We note that unlike the arguments in the upper bound proofs, 
this part of our lower bound technique will work best for random variables of discrete support.}

\begin{theorem}\label{lb:thm-exist}
For all $n, k, m \in \Z_+$ there exists $N \in \Z_+$ such that the following holds:
If there exists an algorithm $A$ that for every pair of distributions $p$ and $q$, supported over $[N]$,
distinguishes the case $p=q$ from the case $\lVert p-q \rVert_{\mathcal{A}_k} \gs \eps$ drawing $m$ samples,
then there exists an algorithm $A'$ that  for every pair of distributions $p'$ and $q'$ supported on $[n]$
distinguishes the case $p'=q'$ versus $\lVert p'-q' \rVert_{\mathcal{A}_k} \gs \eps$ using the same number samples $m$.
Moreover, $A'$ only considers the ordering of the samples and ignores all other information.
\end{theorem}

\begin{proof}
{As a preliminary simplification, we assume that our algorithm, 
instead of taking $m$ samples from any combination of $p$ or $q$ of its choosing, 
takes exactly $m$ samples from $p$ and $m$ samples from $q$, 
as such algorithms are strictly more powerful. 
This also allows us to assume that the algorithm merely takes these random samples 
and applies some processing to determine its output.}

As a critical tool of our proof, we will use the classical Ramsey theorem for hypergraphs.
For completeness, we restate it here in a slightly adapted form.

\begin{lemma}[Ramsey theorem for hypergraphs, \cite{conlon2010hypergraph}]\label{lb:lemma-ramsey}
Given a set $S$ and an integer $t$ let $\binom{S}{t}$ denote the set of subsets of $S$ of cardinality $t$.
For all positive integers, $a$, $b$ and $c$, there exists a positive integer $N$
so that for any function $f:\binom{[N]}{a}\rightarrow [b]$, there exists an $S\subset [N]$ with $|S|=c$ so that $f$ is constant on $\binom{S}{a}$.
\end{lemma}

In words, this means that if we color all subsets of size $a$ of a size $N$ set with at most $b$ different colors, then for large enough $N$ we will find a (bigger) subset $T$ such that all its subsets are colored with the same color. Note that in our setting $c$ from the theorem equals $n$.

The idea of our proof is as follows. Given an algorithm $A$, we will use it to implement the algorithm $A'$. Given $A$, we produce some monotonic function $f:[n]\rightarrow[N]$, and run $A$ on the distributions $f(p)$ and $f(q)$. Since $f$ is order preserving, $\|f(p)-f(q)\|_{\mathcal{A}_k}=\|p-q\|_{\mathcal{A}_k}$, so our algorithm is guaranteed to work. The tricky part will be to guarantee that the output of this new algorithm $A'$ depends only on the ordering of the samples that it takes. Since we may assume that $A$ is deterministic, once we pick which $2m$ samples are taken from $[N]$ the output will be some function of the ordering of these samples (and in particular which are from $p$ and which are from $q$). For the algorithm $A$, this function may depend upon the values that the samples happened to have. Thus, for $A'$ to depend only on order, we need it to be the case that $A$ behaves the same way on any subset of $\textrm{Im}(f)$ of size $2m$. Fortunately, we can find such a set using {Lemma} \ref{lb:lemma-ramsey}.

Since our sample set has size at most $2m$, it is clear that the total number of possible sample sets is at most $N^{2m}$. We color each of these subsets of $[N]$ of size $a=2m$ one of a finite number of colors. The color associates to the sample set the function that $A$ uses to obtain an output given $2m$ samples given by this set coming in a particular order (some of which are potentially equal). The total number of such functions is at most $b=2^{2^{4m}}$. We let $n$ be the proposed support size for $p'$ and $q'$. By Lemma \ref{lb:lemma-ramsey}, for $N$ sufficiently large, there are sets of size $n$ such that the function has the same value in samples from these sets. Letting $f$ be the unique monotonic function from $[n]$ to $[N]$ with this set as its image, causes the output $A'$ to depend only on the ordering of the samples. 

The above reduction works as long as the samples given to our algorithm $A'$ are distinct. 
To deal with the case where samples are potentially non-distinct, we show that it is possible 
to reduce to the case where all $2m$ samples are distinct with $9/10$ probability. 
To do this, we divide each of our original bins into $200m^2$ sub-bins, and upon 
drawing a sample from a given bin, we assign it instead to a uniformly random sub-bin. 
This procedure maintains the $\mathcal{A}_k$ distance between our distributions, 
and guarantees that the probability of a collision is small. 
Now, our algorithm $A'$ will depend only on the order of the samples so long as there is no collision. 
As this happens with probability $9/10$, we can also ensure that this is the case when collisions do occur without sacrificing correctness.
This completes our proof.
\end{proof}

\medskip

We will now give the ``hard'' instance of the testing problem
for algorithms that only consider the ordering of the samples.
We will first describe a construction that works
for $\epsilon =  \Omega(k^{-1/6})$.
We define a mini-bucket to be a segment $I$, which can be divided into three subsegments $I_1,I_2,I_3$
in that order so that $p(I_1)=p(I_3)=\epsilon/(2k)$, $p(I_2)=0$,  and
$q(I_1)=q(I_3)=0,q(I_2)=\epsilon/k$.
We define a bucket to be an interval consisting of a mini-bucket
followed by an interval on which $p=q$ and
on which both $p, q$ have total mass $(1-\epsilon)/k$.
Our distributions for $p$ and $q$ will consist of $k$ consecutive buckets.
See Figure~1 for an illustration.


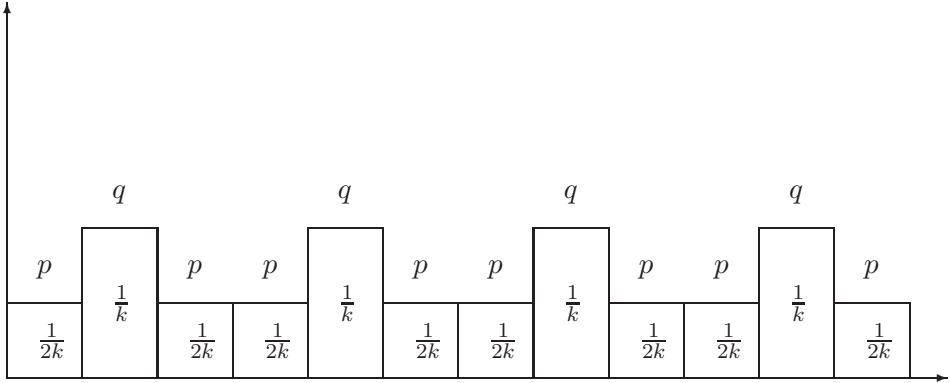
\begin{figure}[h!] \label{fig}
\begin{center}

\setlength{\unitlength}{1cm}
\begin{picture}(13,5)

	
	\put(0,0){\vector(1,0){12.5}}
	\put(0,0){\vector(0,1){5}}
	\multiput(0,0)(1,0){13}%
		{\line(0,1){1}}

	\multiput(0,1)(3,0){4}%
		{\line(1,0){1}}
		
	\multiput(2,1)(3,0){4}%
		{\line(1,0){1}}

	\multiput(1,0)(3,0){4}%
		{\line(0,2){2}}
	\multiput(2,0)(3,0){4}%
		{\line(0,2){2}}

	\multiput(1,2)(3,0){4}%
		{\line(1,0){1}}
	
		
	\multiput(0.4,1.4)(3,0){4}%
		{$p$}
		
	\multiput(2.4,1.4)(3,0){4}%
		{$p$}
		
	\multiput(1.4,2.4)(3,0){4}%
		{$q$}
		
	\multiput(0.4,0.4)(3,0){4}%
		{$\frac{1}{2k}$}
	
	\multiput(2.4,0.4)(3,0){4}%
		{$\frac{1}{2k}$}
		
	\multiput(1.4,0.9)(3,0){4}%
		{$\frac{1}{k}$}

\end{picture}

\end{center}
\caption{$\phi=\frac{1}{2}p + \frac{1}{2}q$ when $\epsilon=1$}\label{fig:p_plus_q}
\end{figure}

Next consider partitioning the domain
into macro-buckets each of which is a union of buckets of total mass $\Theta(1/m)$.
Note that these distributions have $\mathcal{A}_{2k+1}$ distance of $2\epsilon$.
An important fact to note is the following:
\begin{obs} \label{obs}
If zero, one or two draws are made randomly and independently from $(p+q)/2$ on a mini-bucket, then the distribution of which of $p$ or $q$ the samples came from and their relative ordering is indistinguishable from the case where $p=q$.
\end{obs}

To prove the lower bound for the algorithm $A'$,
which is only allowed to look at the ordering of samples. 
We let $X$ be a random variable that is taken to be $0$ or $1$ each with probabilty $1/2.$
When $X=0$ we define $p$ and $q$ as above with mini-buckets,  macro-buckets
and regular buckets as described.
When $X=1$, we let $p=q$ and define mini-buckets
to have total mass $\epsilon/k$ for each of $p$ and $q$,
buckets to have total mass $1/k$ each,
and we combine buckets into macro-buckets as in the $X=0$ case.

Let $Y$ be the distribution on the (ordered) sequences,
obtained by drawing $m'=\Poi(m)$ samples from $p$
and $m''=\Poi(m)$ samples from $q$, with $p$ and $q$ given by $X$.
We are interested in bounding the mutual information between $X$ and $Y$,
since it must be $\Omega(1)$ if the algorithm is going to succeed with probability bounded away from $1/2.$
We show the following:

\begin{theorem}\label{sharedInformationThm}
We have that $I(X:Y) = O(m^5\eps^6/k^4).$
\end{theorem}

\begin{proof}
We begin with a couple of definitions. Let $Y'$ denote $(Y,\alpha)$,
where $\alpha$ is the information about which draws come from which macro-bucket.
$Y'$ consists of $Y'_i$, the sequence of samples coming from the $i$-th macro-bucket.
Note that $$ I(X:Y) \ls I(X:Y') \ls \sum_{i=1}^{O(m)}I(X:Y_i') \;.$$

We will now estimate $I(X:Y_i')$. We claim that it is $O(\frac{m^4\epsilon^6}{k^4})$ for each $i$.
This would cause the sum to be small enough and give our theorem.
We have that, $$I(X:Y_i') = \E_y\left[O\left(1 - \frac{\Pr(Y'_i = y| X=0)}{\Pr(Y'_i = y| X=1)}\right)^2\right].$$
We then have that
$$I(X:Y_i') = \sum_{\ell =0}^{\infty}\sum_{y:|y| = \ell} \frac{O(1)^\ell}{\ell!}O\left(1 - \frac{\Pr(Y'_i = y| X=0,|y|=\ell)}{\Pr(Y'_i = y| X=1,|y|=\ell)}\right)^2.$$
We note that if $X=1,|y|=\ell$ that any of the $2^\ell$ possible orderings are equally likely.
On the other hand, if $X=0$, this also holds in an approximate sense.
To show this, first consider picking which mini-buckets our $\ell$ draws are from.
If no three land in the same mini-bucket, then Observation~\ref{obs} implies that all orderings are equally likely.
Therefore, the statistical distance between $Y'_i|X=0,|y|=\ell$ and $Y'_i|X=1,|y|=\ell$ is at most the probability
that some three draws come from the same mini-bucket.
This is in turn at most the expected number of triples that land in the same mini-bucket,
which is equal to $\binom{\ell}{3}$ times the probability that a particular triple does.
The probability of landing in a particular mini-bucket is $O(m\epsilon/k)^3$.
By definition, there are $O(m/k)$ mini-buckets in a macro-bucket, so this probability is $O(\ell^3 \epsilon^3(m/k)^2)$. Therefore, we have that

\begin{align*}
I(X:Y_i') & = \sum_\ell \frac{O(1)^\ell}{\ell!} \sum_{y:|y|=\ell} O(4^{\ell})\left(\Pr(Y'_i = y| X=0,|y|=\ell) - \Pr(Y'_i = y| X=1,|y|=\ell)\right)^2\\
& \leqslant \sum_\ell \frac{O(1)^\ell}{\ell!} \left(\sum_{y:|y|=\ell}\left|\Pr(Y'_i = y| X=0,|y|=\ell) - \Pr(Y'_i = y| X=1,|y|=\ell) \right|\right)^2\\
& = \sum_\ell \frac{O(1)^\ell}{\ell!}O\left(\ell^6 \epsilon^6m^4/k^4 \right)\\
& = \frac{m^4}{k^4} \sum_\ell \frac{O(1)^\ell \ell^6 \eps^6}{\ell!}\\
& = O\left(\frac{m^4\epsilon^6}{k^4}\right).
\end{align*}
This completes our proof.
\end{proof}

The above construction only works when $k\geqslant m$,
or equivalently, when $\epsilon  = \Omega (k^{-1/6})$. When $\epsilon$ is small,
we need a slightly different construction. We will similarly split our domain into mini-buckets
and macro-buckets and argue based on shared information.
Once again we define two distributions $p$ and $q$, though this time the distributions
themselves will need to be randomized.
Given $k$ and $\epsilon$, we begin by splitting the domain into $k$ macro-buckets.
Each macro-bucket will have mass $1/k$ under both $p$ and $q$.

First pick a global variable $X$ to be either $0$ or $1$ with equal probability.
If $X=1$ then we will have $p=q$ and if $X=0$, $\|p-q\|_{\mathcal{A}_{2k+1}}=\epsilon.$
For each macro-bucket, pick an $x$ uniformly in $[0,(1-\epsilon)/k]$.
The macro-bucket will consist of an interval on which $p=q$ with mass $x$ (for each of $p, q$),
followed by a mini-bucket, followed by an interval of mass $(1-\epsilon)/k-x$ on which $p=q$.
The mini-bucket is an interval of mass $\epsilon/k$ under either $p$ or $q$.
If $X=1$, we have $p=q$ on the mini-bucket.
If $X=0$, the mini-bucket consists of an interval of mass $\epsilon/(2k)$
under $q$ and $0$ under $p$, an interval of mass $\epsilon/k$
under $p$ and $0$ under $q$, and then another interval
of mass $\epsilon/(2k)$ under $q$ and $0$ under $p$.

We let $Y$ be the random variable associated with the ordering of elements
from a set of $\textrm{Poi}(m)$ draws from each of $p$ and $q$. We show:
\begin{theorem}\label{sharedInformationThm2}
If $m \eps =O( k)$, $\log(mk/\eps) =O( \eps^{-1})$, and $k=O(m)$, 
with implied constants sufficiently small,
then $I(X:Y) = O(m^5\eps^6/k^4).$
\end{theorem}

Note that the above statement differs from Theorem \ref{sharedInformationThm} in that $X$ and $Y$ are defined differently.

\begin{proof}
Once again, we let $Y'$ be $Y$ along with the information of which draws came from which macro-bucket, and let $Y_i'$ be the information of the draws from the $i$-th
macro-bucket along with their ordering. It suffices for us to show that $I(X:Y_i')=O(m^5\epsilon^6k^{-5})$ for each $i$ (as now there are only $k$ macro-buckets rather than $m$).

Let $s$ be a string of $\ell$ ordered draws from $p$ and $q$.
In particular, we may consider $s$ to be a string $s_1s_2\ldots s_\ell$,
where $s_i\in\{ p, q \}$. We wish to consider the probability that $Y_i'=s$
under the conditions that $X=0$ or that $X=1$. In order to do this, we further
condition on which elements of $s$ were drawn from the mini-bucket.
For $1\leqslant a\leqslant b\leqslant \ell$ we consider the probability that not only
did we obtain sequence $s$, but that the draws $s_a,\ldots,s_b$
were exactly the ones coming from the mini-bucket within this macro-bucket.
Let $h$ denote the ordered string coming from elements drawn from the mini-bucket
and $M$ the ordered sequence of strings coming from elements not drawn from the mini-bucket.
The probability of the event in question is then
$$
\Pr(h=s_a\ldots s_b)\Pr(M=s_1\ldots s_{a-1}s_{b+1}\ldots s_\ell)\Pr(\textrm{the mini-bucket is placed between }s_{a-1}\textrm{ and }s_{b+1}).
$$
Note that the mini-bucket can be thought of as being randomly and uniformly inserted
within an interval of length $(1-\epsilon)/k$
and that this is equally likely to be inserted between any pair of elements of $M$.
Thus, the probability of the third term in the product is exactly $1/(\ell+a-b)$.
The second probability is the probability that $\ell+a-b-1$ elements are drawn
from the complement of the mini-bucket times $2^{-(\ell+a-b+1)}$,
as draws from $p$ and $q$ are equally likely.
Thus, letting $t=b-a+1$ (i.e., the number of elements in the mini-bucket), we have that
\begin{align*}
\Pr(Y_i'=s) = e^{-m/k}\sum_{t=0}^\ell \left(\frac{\left(\frac{m(1-\epsilon)}{2k} \right)^{\ell-t}}{(\ell-t)!}\right)\left(\frac{\left(\frac{m\epsilon}{k} \right)^t}{t!} \right)\left(\frac{1}{\ell-t} \right)\sum_a \Pr(h=s_a\ldots s_{a+t-1}:|h|=t).
\end{align*}
Note that this equality holds even after conditioning upon $X$. We next simplify this expression further by grouping together terms in the last sum based upon the value of the substring $s_a\ldots s_{a+t-1}$, which we call $r$. We get that
\begin{align*}
\Pr(Y_i'=s) = e^{-m/k}\sum_{t=0}^\ell \left(\frac{\left(\frac{m(1-\epsilon)}{2k} \right)^{\ell-t}}{(\ell-t)!}\right)\left(\frac{\left(\frac{m\epsilon}{k} \right)^t}{t!} \right)\left(\frac{1}{\ell-t} \right)\sum_{|r|=t} \Pr(h=r:|h|=t)N_{r,s},
\end{align*}
where $N_{r,s}$ is the number of occurrences of $r$ as a substring of $s$.

Next, we wish to bound
\begin{equation}\label{condTotVarEq}
\sum_{|s|=\ell}|\Pr(Y_i'=s:X=0)-\Pr(Y_i'=s:X=1)|^2.
\end{equation}
By the above formula this is at most
\begin{eqnarray*}
e^{-2m/k}\sum_{|s|=\ell} \sum_{t=0}^\ell
\left(\frac{\left(\frac{m(1-\epsilon)}{2k} \right)^{\ell-t}}{(\ell-t)!}\right)
\left(\frac{\left(\frac{m\epsilon}{k} \right)^t}{t!} \right)
\left(\frac{1}{\ell-t} \right) \cdot \\
\left|\sum_{|r|=t} N_{r,s}\left(\Pr(h=r:|h|=t,X=0)- \Pr(h=r:|h|=t,X=1)\right)\right|^2.
\end{eqnarray*}
For fixed values of $t$ we consider the sum
$$
\sum_{|s|=\ell}\left|\sum_{|r|=t} N_{r,s}(\Pr(h=r:|h|=t,X=0)-\Pr(h=r:|h|=t,X=1))\right|^2.
$$
Note that if $t\leqslant 2$ then $\Pr(h=r:|h|=t,X=0)=\Pr(h=r:|h|=t,X=1)$, and so the above sum is $0$. Otherwise, it is at most
$$
\sum_{|s|=\ell}\sum_{|r|=t}|N_{r,s}-(\ell+1-t)/2^t|^2
$$
because $\sum_r \Pr(h=r:|h|=t,X=0) = \sum_r \Pr(h=r:|h|=t,X=1) = 1.$ Note on the other hand that the expectation over random strings $s$ of length $\ell$ of $N_{r,s}-(\ell+1-t)/2^t$ is $0$. Furthermore, the variance of $N_{r,s}$ is easily bounded by $t\ell 2^{-t}$ as whether or not two disjoint substrings of $s$ are equal to $r$ are independent events. Therefore, the above sum is at most
$$
2^\ell 2^t {t\ell 2^{-t}} = 2^\ell t \ell.
$$
Hence, by Cauchy-Schwartz, we have that
$$
\sum_{|s|=\ell}\left|\sum_{|r|=t}N_{r,s}-(\ell+1-t)/2^t\right|^2 \leqslant 2^\ell2^t t\ell.
$$

Therefore, the expression in \eqref{condTotVarEq} is at most
$$
e^{-2m/k}\left(\sum_{t=3}^\ell\left(\frac{\left(\frac{m(1-\epsilon)}{2k} \right)^{\ell-t}}{(\ell-t)!}\right)\left(\frac{O\left(\frac{m\epsilon}{k} \right)^t}{t!} \right)\left(\frac{1}{\ell-t} \right)\left(2^\ell 2^t \ell t\right)^{1/2}\right)^2.
$$
Assuming that $\ell\epsilon$ is sufficiently small, these terms are decreasing exponentially with $t$, and thus this is
$$
O\left(e^{-2m/k}\left(\frac{(m^2/(2k^2))^{\ell}}{(\ell!)^2}\right)\epsilon^6\ell^5 \right).
$$
Now we have that for $N$ a sufficiently small constant times $\eps^{-1}$,
\begin{align*}
I(X:Y_i') & = \sum_s\Pr(Y_i'=s:X=1) O\left(1-\frac{\Pr(Y_i'=s:X=0)}{\Pr(Y_i'=s:X=1)} \right)^2\\
& = \sum_{\ell} \sum_{s:|s|=\ell} e^{m/k} \left(\frac{(m/(2k))^{\ell}}{\ell!}\right)^{-1}O(\Pr(Y_i'=s:X=1)-\Pr(Y_i'=s:X=0))^2\\
& \leqslant  \sum_\ell e^{m/k}\left(\frac{(m/(2k))^{\ell}}{\ell!}\right)^{-1} O\left( \sum_{s:|s|=\ell}|\Pr(Y_i'=s:X=1)-\Pr(Y_i'=s:X=0)|^2\right)\\
& \leqslant \sum_{\ell>N} O\left(\frac{(2m/k)^{\ell}}{\ell!}\right) + \sum_{\ell<N}e^{m/k}\left(\frac{(m/(2k))^{\ell}}{\ell!}\right)^{-1}O\left(e^{-2m/k}\left(\frac{(m^2/(2k^2))^{\ell}}{(\ell!)^2}\right)\epsilon^6\ell^5 \right)\\
& \leqslant \sum_{\ell>N} O\left(\frac{m}{kN}\right)^\ell + \sum_{\ell} O\left(e^{-m/k}\frac{(m/k)^\ell}{\ell!}\epsilon^6\ell^5 \right).
\end{align*}
Since $\frac{m}{kN} \leq \frac{m\epsilon}{k}$ is sufficiently small, 
the first term is at most $(1/2)^N$ which is polynomially small in $mk/\eps$, and thus negligible. 
The second term is the expectation of $\eps^6\ell^5$ for $\ell$ a Poisson random variable with mean $m/k$. 
Thus, it is easily seen to be $O((m/k)^5\eps^6)$. Therefore, we have that $I(X:Y_i')=O(m^5\eps^6k^{-5})$, 
and therefore, $I(X:Y)=O(m^5\eps^6k^{-4})$, as desired.
\end{proof}

We are now ready to complete the proof of our general lower bound.
\begin{theorem}
For any $k>2$, there exists an $N$ so that any algorithm 
that is given sample access to two distributions,
$p$ and $q$ over $[N]$, and can distinguish between the cases $p=q$ and $\|p-q\|_{\mathcal{A}_k}$
with probability at least $2/3$, 
requires at least $\Omega\left(\max\left\{k^{4/5}/\epsilon^{6/5},k^{1/2}/\epsilon^2\right\}\right)$ samples.
\end{theorem}
\begin{proof}
The lower bound of $k^{1/2}/\epsilon^2$ follows from the known lower bound~\cite{Paninski:08} even in the case
where $q$ is known and $p$ and $q$ have support of size $k$.
It now suffices to consider the case that $\epsilon > k^{-1/2}$ and $m$ a sufficiently small constant times $k^{4/5}\epsilon^{-6/5}.$

Note that by Theorem \ref{lb:thm-exist}, we may assume that the algorithm in question takes $m$ samples from each of $p$ and $q$ and determines its output
based only on the ordering of the samples. We need to show that this is impossible for $N$ sufficiently large.

We note that if we allow $p$ and $q$ to be continuous distributions instead of discrete ones we are already done. If $m<k$, we use our first counter-example construction,
and if $m\geqslant k$ use the second one. If we let $X$ be randomly $0$ or $1$, and set $p=q$ for $X=1$ and $p,q$ as described above when $X=0$,
then by Theorems \ref{sharedInformationThm} and \ref{sharedInformationThm2}, the shared information between $X$ and the output of our algorithm
is at most $O(m^5 \epsilon^6 k^{-4}) = o(1)$, and thus our algorithm cannot correctly determine $X$ with constant probability.

In order to prove our Theorem, we will need to make this work for distributions $p$ and $q$ with finite support size as follows:
By splitting our domain into $m^3$ intervals each of equal mass under $p+q$, we note that the $\mathcal{A}_k$ distance between the distributions is only negligibly affected. Furthermore, with high probability, $m$ samples will have no pair chosen from the same bin. Thus, the distribution on orderings of samples from these discrete distributions are nearly identical to the continuous case, and thus our algorithm would behave nearly identically. This completes the proof.
\end{proof}

\bibliographystyle{alpha}

\nocite{}

\bibliography{allrefs}

\appendix

\section*{Appendix}

\section{Proof of Proposition~\ref{UniformTesterProp}} \label{sec:app}

In this section, we prove Proposition \ref{UniformTesterProp}.
We note that it suffices to attain confidence probability
$2/3$ with $O(\sqrt{k}/\epsilon^2)$ samples, as we can then run $O(\log(1/\delta))$ independent iterations
to boost the confidence to $1-\delta.$
Our starting point is the following Theorem from~\cite{CDVV14}:

\begin{theorem} [\cite{CDVV14}, Proposition 3.1] \label{thm:l2-unif}
For any distributions $p$ and $q$ over $[n]$ such
that $\lVert p \rVert_2 \leqslant \frac{O(1)}{\sqrt{n}}$ and $\lVert q \rVert_2 \leqslant \frac{O(1)}{\sqrt{n}}$
there is a testing algorithm that distinguishes with probability at least $2/3$ the case that $q = p$
from the case that $||q-p||_2 \gs \eps/\sqrt{n}$ when given $O(\sqrt{n}/\eps^2)$
samples from $q$ and $p$.
\end{theorem}

Our $\mathcal{A}_k$ testing algorithm for this regime is the following:



\medskip

\fbox{\parbox{6.2in}{
{\bf Algorithm} {Test-Identity-Flat-$\mathcal{A}_k(p, q, n, \eps)$\\}
Input: sample access to distributions $p$ and $q$ over $[n]$ 
with $\|p\|_2,\|q\|_2=O(1/\sqrt{n})$, $k \in \Z_+$ with $2 \le k \le n$, and $\eps>0$.\\
Output: ``YES'' if $q = p$; ``NO'' if $\|q-p\|_{\mathcal{A}_k} \ge \eps.$

\begin{enumerate}

\item Draw samples {$S_1$, $S_2$} of size $m = O(\sqrt{k}/\eps^2)$ from {$q$ and $p$}.

\item By artificially increasing the support if necessary, we can guarantee that $n = k \cdot 2^{j_0}$,
where $j_0  \eqdef \lceil \log_2(1/\eps) \rceil+O(1).$

\item 
Consider the collection
$\{\mathcal{I}^{(j)}\}_{j=0}^{j_0-1}$ of $j_0$ partitions of $[n]$ into intervals;
the partition $\mathcal{I}^{(j)} = (I_i^{(j)})_{i=1}^{\ell_j}$ consists of $\ell_j = k \cdot 2^j$ many intervals with
 $I_i^{(j_0)}$ of length $n/\ell_j+O(1)$, and $I_i^{(j)}$ the union of two adjacent intervals of $I_i^{(j+1)}$.

\vspace{-0.2cm}

  \item For $j=0, 1, \ldots, j_0-1$:
  \vspace{-0.2cm}
	\begin{enumerate}
		\item Consider the reduced distributions $q_r^{\mathcal{I}^{(j)}}$ and $p_r^{\mathcal{I}^{(j)}}$.
		          Use {the samples $S_1$, $S_2$} to simulate samples to { $q_r^{\mathcal{I}^{(j)}}$ and $p_r^{\mathcal{I}^{(j)}}$ }.
		          \vspace{-0.2cm}
		\item  Run {Test-Identity-${L_2}(q_r^{\mathcal{I}^{(j)}}, p_r^{\mathcal{I}^{(j)}}, \ell_j, \eps_j, \delta_j)$}
		for $\eps_j = C \cdot \eps \cdot 2^{3j/8}$ for $C>0$ a sufficiently small constant and $\delta_j =  2^{-j}/6$,
		i.e., test whether { $q_r^{\mathcal{I}^{(j)}} = p_r^{\mathcal{I}^{(j)}}$} versus 
		{$\| q_r^{\mathcal{I}^{(j)}} - p_r^{\mathcal{I}^{(j)}}\|_2 > \gamma_j \eqdef \eps_j / \sqrt{\ell_j}.$}
	\end{enumerate}
	\vspace{-0.2cm}
  \item If all the testers in Step~3(b) output  ``YES'', then output ``YES'';  otherwise output ``NO''.
\end{enumerate}
\vspace{-0.3cm}
}}
\smallskip

Note in the above that when $\epsilon_j>1$, that the appropriate tester requires no samples.
The following proposition characterizes the performance of the above algorithm.

\begin{proposition}\label{prop:ak}
The algorithm {\em Test-Identity-Flat-}$\mathcal{A}_k(p, q, n, \eps)$, on input a sample of size $m = O(\sqrt{k}/\eps^2)$
drawn from distributions $q$ and $p$ over $[n]$ with $\|p\|_2,\|q\|_2=O(1/\sqrt{n})$,
$\eps>0$, and an integer $k$ with $2 \le k \le n$, correctly distinguishes the case that $q = p$
from the case that $\|q-p\|_{\mathcal{A}_k} \ge \eps$, with probability at least $2/3$.
\end{proposition}
\begin{proof}
First, it is straightforward to verify the claimed sample complexity, as the algorithm only draws samples in Step~1.
Note that the algorithm uses the same set of samples $S_1, S_2$ for all testers in Step~4(b).
Note that it is easy to see that $\|p_r^{\mathcal{I}^{(j)}}\|_2,\|q_r^{\mathcal{I}^{(j)}}\|_2=O(1/\sqrt{\ell_j})$, and therefore,
by Theorem~\ref{thm:l2-unif}, the tester Test-Identity-${L_2}(q_r^{\mathcal{I}^{(j)}}, p_r^{\mathcal{I}^{(j)}}, \ell_j, \eps_j, \delta_j)$,
on input a set of $m_j = O((\sqrt{\ell_j}/\eps_j^2) \cdot \log(1/\delta_j))$ samples from $q_r^{\mathcal{I}^{(j)}}$ and
$p_r^{\mathcal{I}^{(j)}}$
distinguishes the case
that $q_r^{\mathcal{I}^{(j)}} = p_r^{\mathcal{I}^{(j)}}$ from the case that  $\|q_r^{\mathcal{I}^{(j)}} - p_r^{\mathcal{I}^{(j)}}\|_2
\ge \gamma_j \eqdef \eps_j / \sqrt{\ell_j}$ with probability at least
$1-\delta_j$. From our choice of parameters it can be verified that $\max_j m_j \le m= O(\sqrt{k}/\eps^2)$, hence we can use the same sample $S_1, S_2$
as input to these testers for all $0 \le j \le j_0-1$. In fact, it is easy to see that $\littlesum_{j=0}^{j_0-1} m_j = O(m)$,
which implies that the overall algorithm runs in sample-linear time.
Since each tester in Step 3(b) has error probability $\delta_j$, by a union bound over all $j \in \{0, \ldots, j_0-1\}$, the total error probability is at most
$\littlesum_{j=0}^{j_0-1} \delta_j  \le (1/6) \cdot \littlesum_{j=0}^{\infty} 2^{-j} = 1/3.$
Therefore, with probability at least $2/3$ all the testers in Step~4(b) succeed.
We will henceforth condition on this ``good'' event, and establish the completeness and soundness properties of
the overall algorithm under this conditioning.


We start by establishing completeness. If $q = p$, then for any partition $\mathcal{I}^{(j)}$, $0 \le j \le j_0-1$, we have that
$q_r^{\mathcal{I}^{(j)}} = p_r^{\mathcal{I}^{(j)}} $. By our aforementioned conditioning, all testers in Step~3(b) will output ``YES'',
hence the overall algorithm will also output ``YES'', as desired.

We now proceed to establish the soundness of our algorithm.
Assuming that $\|q - p\|_{\mathcal{A}_k} \ge \eps$, we want to show that the algorithm
Test-Identity-$\mathcal{A}_k(q, n, \eps)$ outputs ``NO'' with probability at least $2/3$.
Towards this end, we prove the following structural lemma:
\begin{lemma} \label{lem:structural} For $C>0$ a sufficiently small constant,
if $\|q - p\|_{\mathcal{A}_k} \ge \eps$, there exists $j \in \Z_+$ with $0 \le j \le j_0-1$ such that
$\| q_r^{\mathcal{I}^{(j)}} - p_r^{\mathcal{I}^{(j)}}\|^2_2 \ge \gamma_j^2.$
\end{lemma}
Given the lemma, the soundness property of our algorithm follows easily.
Indeed, since all testers Test-{Identity}-${L_2}(q_r^{\mathcal{I}^{(j)}}, \ell_j, \eps_j, \delta_j)$ of Step~4(b) are successful by our conditioning,
Lemma \ref{lem:structural} implies that at least one of them outputs ``NO'', hence the overall algorithm will output ``NO''.
\end{proof}

\noindent The proof of Lemma \ref{lem:structural} is very similar to the analogous lemma in~\cite{DKN15a}.
For the same of completeness, it is given in the following subsection.

\subsection{Proof of Lemma~\ref{lem:structural}} \label{sec:ak-proof}

We claim that it is sufficient to take $C\leqslant 5 \cdot 10^{-6}$.
Thus, we are in the case where $n=2^{j_0-1}\cdot k$ and have
argued that it suffices to show that our algorithm works to distinguish 
$\mathcal{A}_k$-distance in this setting with $\eps_j= 10^{-5} \cdot \eps \cdot 2^{3j/8}$.

We make use of the following definition:
\begin{definition}
For $p$ and $q$ arbitrary distributions over $[n]$, we
define the {\em scale-sensitive-$L_2$ distance} between $q$ and $p$ to be
$$
\|q - p\|^2_{[k]} \eqdef \max_{\mathcal{I} = (I_1, \ldots, I_r) \in \mathbf{W}_{1/k}} \sum_{i=1}^r \frac{\discr^2(I_i)}{\width^{1/8}(I_i)}
$$
where $ \mathbf{W}_{1/k}$ is the collection of all interval partitions of $[n]$ into intervals of width at most $1/k$,
$
\discr(I)=|p(I)-q(I)|,
$ and $\width(I)$ is the number of bins in $I$ divided by $n$.
\end{definition}


The first thing we need to show is that if $q$ and $p$ have large $\mathcal{A}_k$ distance then they also have large scale-sensitive-$L_2$ distance.
Indeed, we have the following lemma:
\begin{lemma}\label{AkScaledL2Lem}
For $p$ and $q$ an arbitrary distributions over $[n]$, we have that
$$
\|q-p\|^2_{[k]} \geqslant \frac{\|q-p\|_{\mathcal{A}_k}^2}{(2k)^{7/8}}.
$$
\end{lemma}
\begin{proof}
Let $\eps=\|q - p\|_{\mathcal{A}_k}^2$. Consider the optimal $\mathcal{I}^{\ast}$ in the definition of the $\mathcal{A}_k$ distance. By further subdividing intervals of width more than $1/k$ into smaller ones, we can obtain a new partition, $\mathcal{I}' = (I_i')_{i=1}^{s}$, of cardinality $s \le 2k$ all of whose parts have width at most $1/k$. Furthermore, we have that $\sum_i \discr(I'_i) \geqslant \epsilon$. Using this partition to bound from below $\|q-p\|^2_{[k]}$, by Cauchy-Schwarz we obtain that
\begin{align*}
\|q-p\|^2_{[k]} & \geqslant \sum_i \frac{\discr^2(I'_i)}{\width(I'_i)^{1/8}}\\
& \geqslant \frac{\left(\sum_i \discr(I'_i) \right)^2}{\sum_i \width(I'_i)^{1/8}}\\
& \geqslant \frac{\epsilon^2}{2k (1/(2k))^{1/8}} \\ & = \frac{\epsilon^2}{(2k)^{7/8}}.
\end{align*}
\end{proof}

The second important fact about the scale-sensitive-$L_2$ distance is that if it is large then one of the partitions considered in our algorithm will produce a large $L_2$ error.
\begin{proposition}\label{scaledL2IntProp}
Let $p$ and $q$ be distributions over $[n]$. Then we have that
\begin{equation}\label{scaledDistErrsEqn}
\|q-p\|_{[k]}^2 \leqslant 10^8 \sum_{j=0}^{j_0-1} \sum_{i=1}^{2^j \cdot k} \frac{\discr^2(I_i^{(j)})}{\width^{1/8}(I_i^{(j)})}.
\end{equation}
\end{proposition}
\begin{proof}
Let $\mathcal{J}  \in \mathbf{W}_{1/k}$ be the optimal partition used when computing the scale-sensitive-$L_2$ distance $\|q-p\|_{[k]}$. In particular, it is a partition into intervals of width at most $1/k$ so that $\sum_i \frac{\discr^2(J_i)}{\width(J_i)^{1/8}}$ is as large as possible. To prove \eqref{scaledDistErrsEqn}, we prove a notably stronger claim. In particular, we will prove that for each interval $J_\ell\in\mathcal{J}$
\begin{equation}\label{refinedSumEqn}
\frac{\discr^2(J_\ell)}{\width^{1/8}(J_\ell)} \leqslant 10^8\sum_{j=0}^{j_0-1} \sum_{i: I_i^{(j)}\subset J_\ell} \frac{\discr^2(I_i^{(j)})}{\width^{1/8}(I_i^{(j)})}.
\end{equation}
Summing over $\ell$ would then yield $\|q-p\|_{[k]}^2$ on the left hand side and a strict subset of the terms from \eqref{scaledDistErrsEqn} on the right hand side. From here on, we will consider only a single interval $J_\ell$. For notational convenience, we will drop the subscript and merely call it $J$.

First, note that if $|J|\leqslant 10^8$, then this follows easily from considering just the sum over $j=j_0-1$. Then, if $t=|J|$, $J$ is divided into $t$ intervals of size $1$. The sum of the discrepancies of these intervals equals the discrepancy of $J$, and thus, the sum of the squares of the discrepancies is at least $\discr^2(J)/t$. Furthermore, the widths of these subintervals are all smaller than the width of $J$ by a factor of $t$. Thus, in this case the sum of the right hand side of \eqref{refinedSumEqn} is at least $1/t^{7/8}\geqslant \frac{1}{10^7}$ of the left hand side.

Otherwise, if $|J|>10^8$, we can find a $j$ so that $\width(J)/10^8 < 1/(2^j\cdot k) \leqslant 2\cdot \width(J)/10^8$. We claim that in this case Equation \eqref{refinedSumEqn} holds even if we restrict the sum on the right hand side to this value of $j$. Note that $J$ contains at most $10^8$ intervals of $\mathcal{I}^{(j)}$, and that it is covered by these intervals plus two narrower intervals on the ends. Call these end-intervals $R_1$ and $R_2$. We claim that $\discr(R_i)\leqslant \discr(J)/3$. This is because otherwise it would be the case that
$$
\frac{\discr^2(R_i)}{\width^{1/8}(R_i)} > \frac{\discr^2(J)}{\width^{1/8}(J)}.
$$
(This is because $(1/3)^2\cdot (2/10^8)^{-1/8} > 1$.)
This is a contradiction, since it would mean that partitioning $J$ into $R_i$ and its complement would improve the sum defining $\|q-p\|_{[k]}$, which was assumed to be maximum. This means that the sum of the discrepancies of the $I_i^{(j)}$ contained in $J$ must be at least $\discr(J)/3$, so the sum of their squares is at least $\discr^2(J)/(9\cdot 10^8)$. On the other hand, each of these intervals is narrower than $J$ by a factor of at least $10^8/2$, thus the appropriate sum of $\frac{\discr^2(I_i^{(j)})}{\width^{1/8}(I_i^{(j)})}$ is at least $\frac{\discr^2(J)}{10^8\width^{1/8}(J)}$. This completes the proof.
\end{proof}

We are now ready to prove Lemma \ref{lem:structural}.
\begin{proof}
If $\|q-p\|_{\mathcal{A}_k}\geqslant \eps$ we have by Lemma \ref{AkScaledL2Lem} that
$$\|q-p\|_{[k]}^2 \geqslant \frac{\epsilon^2}{(2k)^{7/8}}.$$ By Proposition \ref{scaledL2IntProp}, this implies that
\begin{align*}
\frac{\eps^2}{(2k)^{7/8}} &\leqslant 10^8\sum_{j=0}^{j_0-1} \sum_{i=1}^{2^j\cdot k} \frac{\discr^2(I_i^{(j)})}{\width^{1/8}(I_i^{(j)})}\\
& = 10^8 \sum_{j=0}^{j_0-1} (2^{j}k)^{1/8} \|q^{\mathcal{I}^{(j)}}-{p^{\mathcal{I}^{(j)}}}\|_2^2.
\end{align*}
Therefore,
\begin{equation}\label{L2SumEqn}
\sum_{j=0}^{j_0-1} 2^{j/8} \|q^{\mathcal{I}^{(j)}}-{p^{\mathcal{I}^{(j)}}}\|_2^2 \geqslant 5\cdot 10^{-9} \eps^2/k.
\end{equation}
On the other hand, if $\|q^{\mathcal{I}^{(j)}}-{p^{\mathcal{I}^{(j)}}}\|_2^2$ were at most $10^{-10}2^{-j/4}\epsilon^2/k$ for each $j$, then the sum above would be at most
$$
10^{-10}\epsilon^2/k \sum_j 2^{-j/8} < 5\cdot 10^{-9} \eps^2/k.
$$
This would contradict Equation \eqref{L2SumEqn}, thus proving that $\|q^{\mathcal{I}^{(j)}}-U_{\ell_j}\|_2^2\geqslant10^{-10}2^{-j/4}\epsilon^2/k$ for at least one $j$,
proving Lemma~ \ref{lem:structural}.
\end{proof}

\end{document}